\documentclass{fundam}

\usepackage{url} 
\usepackage[ruled,lined]{algorithm2e}
\usepackage{graphicx}
\usepackage{tikz}
\usetikzlibrary{automata, positioning, arrows}
\usepackage{amsmath,amssymb,amsfonts,array}

\begin{document}


\setcounter{page}{185}
\publyear{22}
\papernumber{2107}
\volume{185}
\issue{2}

  \finalVersionForIOS


\title{On the $2$-domination Number of Cylinders with Small Cycles}

\author{Ester M. Garz\'on,$\,$ Jos\'e A. Mart\'inez,$\,$ Juan J. Moreno   \\
Department of Computer Sciences and Agrifood Campus of International Excellence (ceiA3) \\
Universidad de Almer\'ia \\
Carretera Sacramento s/n, 04120 Almer\'ia, Spain\\
gmartin@ual.es, jmartine@ual.es, juanjomoreno@ual.es
\and  Mar\'ia L. Puertas\thanks{Address for correspondence: Department of Mathematics and Agrifood Campus of
                   International Excellence (ceiA3), Universidad de Almer\'ia, Carretera Sacramento s/n, 04120 Almer\'ia, Spain. \newline \newline
          \vspace*{-6mm}{\scriptsize{Received September 2021; \ accepted February 2022.}}}
  \\
Department of Mathematics and Agrifood Campus of International Excellence (ceiA3) \\
Universidad de Almer\'ia \\
Carretera Sacramento s/n, 04120 Almer\'ia, Spain\\
mpuertas@ual.es
}

\maketitle

\runninghead{E.M. Garz\'on et al.}{On the $2$-domination Number of Cylinders with Small Cycles}

\begin{abstract}
Domination-type parameters are difficult to manage in Cartesian product graphs and there is usually no general relationship between the parameter in both factors and in the product graph. This is the situation of the domination number, the Roman domination number or the $2$-domination number, among others. Contrary to what happens with the domination number and the Roman domination number, the $2$-domination number remains unknown in cylinders, that is, the Cartesian product of a cycle and a path and in this paper, we will compute this parameter in the cylinders with small cycles. We will develop two algorithms involving the $(\min,+)$ matrix product that will allow us to compute the desired values of $\gamma_2(C_n\Box P_m)$, with $3\leq n\leq 15$ and $m\geq 2$. We will also pose a conjecture about the general formul\ae\ for the $2$-domination number in this graph class.
\end{abstract}


\section{Introduction}\label{Sec:Introduction}

Domination-type parameters in graphs (see~\cite{Haynes1998}) are a well-known tool to approach the problem of efficiently locating resources in a network. The variety of such parameters allows us to address different distribution requirements of such resources. The original parameter is the domination number. A dominating set in a graph $G$ is a vertex set $D$ such that every vertex not in $D$ has at least one neighbor in it. The domination number of $G$ is $\gamma (G)$, the minimum cardinal of a dominating set of $G$. This parameter is still under study (see~\cite{Bujtas2021}). Moreover, several variations of it have been defined, even quite recently (see~\cite{Chellali2021}). Applications of some such variations to the optimal location of radio stations or land surveying sensors can be found in~\cite{Haynes1998}.

If the distribution requirement is that every node of the network should have access to at least two resources, the so-called $2$-domination arises (see~\cite{Fink1985}). A $2$-dominating set of a graph $G$ is a vertex subset $S$ such that every vertex not in $S$ has at least two neighbors in $S$. The $2$-domination number $\gamma_2(G)$ is the minimum cardinal of a $2$-dominating set of $G$.

\medskip
Following the general trend of domination-type parameters, the problem of computing the $2$-domination number is NP-complete in general graphs (see~\cite{Bonomo2018}). Furthermore, these parameters are difficult to compute in Cartesian product graphs. Exhaustive information about graph products can be found in~\cite{Imrich2000}. For information about domination in Cartesian product graphs see~\cite{Vizing1968,Bresar2012,Bresar2021}. Recall that the Cartesian product of two graphs $G\Box H$ is defined as follows:
\begin{itemize}
\item the vertex set is $V(G\Box H)=V(G)\times V(H)$, that is, the Cartesian product of the sets $V(G)$ and $V(H)$,
\item two vertices $(g_1,h_1), (g_2,h_2)$ are adjacent in $G\Box H$ if and only if
    \begin {itemize}
    \item either $g_1=g_2$ and $h_1, h_2$ are adjacent in $H$,
    \item or $g_1,g_2$ are adjacent in $G$ and $h_1=h_2$.
    \end{itemize}
\end{itemize}

The $2$-domination number remains unknown for general Cartesian product graphs while the $2$-domination number of particular cases of Cartesian product of two paths~\cite{Rao2019} or the Cartesian product of two cycles~\cite{Shaheen2009} have been computed. Moreover, the general problem is still open for the Cartesian product of a path and a cycle and the Cartesian product of two cycles, but the $2$-domination number of the Cartesian product of two paths has recently been obtained in~\cite{Rao2019}, by using appropriate algorithms. The main tool of these algorithms is the $(\min,+)$ matrix multiplication, defined over the semi-ring $\mathcal{P}=(\mathbb{R}\cup\{\infty\}, \min, +, \infty, 0)$ of tropical numbers in the minimum convention~\cite{Pin98}. This matrix multiplication, that we denote with $\boxtimes$, is defined by $A \boxtimes B=(c_{ij}=\min_{k}(a_{ik}+b_{kj}))$, that is, the operations multiplication and addition in the usual matrix product are replaced by addition and minimization, respectively. Moreover, the $(\min,+)$ product of a matrix $A$ and $\alpha\in \mathbb{R}\cup\{\infty\}$ is defined by $(\alpha\boxtimes A)_{ij}=\alpha+a_{ij}$.

Similar algorithms have been used to approach the computation of several domination-type parameters in Cartesian product graphs. The technique was originally presented in~\cite{Klavzar1996} for fasciagraphs and rotagraphs, in which Cartesian products of paths and cycles are particular cases, and it also appears in~\cite{Rao2019,Spalding1998,Goncalves2011,Pavlic2012,Pavlic2013}, among others. We follow these ideas in this paper.

In the cases mentioned above, the process has two steps. Firstly, it is necessary to compute the value of the parameter in some ``small'' cases consisting of bounding the order of one of the factors of the Cartesian product graph. The behavior of the domination-like parameters in such graphs is not regular for very small paths or cycles, but it becomes regular for big enough cases. Once such regular behavior becomes apparent, a specific procedure can be designed to obtain the value of the desired parameter for the general case.

In grids $P_m\Box P_n$ both factors are paths and bounding the order of any of them is equal. However, there are two options in cylinders $C_n\Box P_m$, either bounding the cycle order or the path order, that is, considering either the case $3\leq n\leq N, m\geq 2$, or the case $2\leq m\leq M, n\geq 3$. In this paper we will focus on the first case and we compute the $2$-domination number of cylinders with a small cycle, by means of two algorithms that use the $(\min,+)$ product of  large sparse matrices and dense vectors. In Section~\ref{Sec:theory} we present the theoretical results needed to ensure the validity of the algorithms that we run in Section~\ref{Sec:algorithms}. Finally, we sum up our results in Section~\ref{sec:conclusions}.

\section{Theoretical results}\label{Sec:theory}

In this section we present the results that will allow us to design two algorithms to compute the $2$-domination number of selected cylinders with small cycles and any path.

\medskip
We will use the following notation for the cylinder $C_n\Box P_m$. The vertex set is $V(C_n\Box P_m)=\{ v_{ij}\colon 0\leq i\leq n-1, 0\leq j\leq m-1\}$. The $i$-th row is the subgraph generated by the vertex subset $\{v_{ij}\colon 0\leq j\leq n-1\}$, that is isomorphic to $P_m$, and the $j$-th column is the subgraph generated by $\{v_{ij}\colon 0\leq i\leq m-1\}$, being isomorphic to $C_n$. We numerate columns from left to right (see Figure~\ref{Figure:rows}).

\begin{figure}[h]
\vspace{1mm}
\centerline{\includegraphics[width=0.3\textwidth]{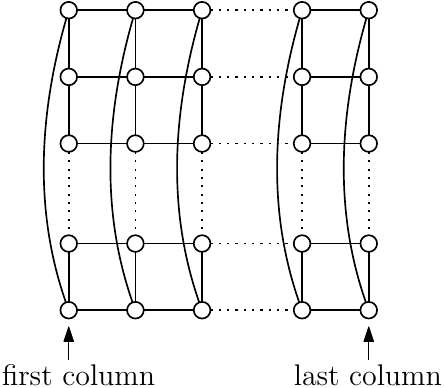}}
\caption{The columns of the cylinder $C_n\Box P_m$.\label{Figure:rows}}
\end{figure}

The neighbors of a vertex $v_{ij}$ are the following (the first index is taken module $n$):
\begin{itemize}
\itemsep=0.95pt
\item if $j=1$, then $v_{ij}=v_{i1}$ is on the first column and its neighbors are $v_{(i-1) 1}, v_{(i+1) 1}, v_{i2}$,

\item if $j=m-1$, then $v_{ij}=v_{i(m-1)}$ is on the last column and its neighbors are $v_{(i-1)(m-1)},$ $v_{(i+1)(m-1)},$ $v_{i(m-2)}$,

\item if $1<j<m-1$, then $v_{ij}$ has four neighbors: $v_{(i-1) j}, v_{(i+1) j}, v_{i (j-1)}, v_{i (j+1)}$.
\end{itemize}


The general idea of the algorithm is to run the same routine several times, in such way that the first step includes only the first column and the successive steps add one column at a time. To this end, we need the following definition similar to the $2$-dominating set but leaving the last column not dominated, bearing in mind that such column could be dominated in the following step.

\begin{definition}
A quasi-$2$-dominating set of the cylinder $C_n\Box P_m$ is a vertex subset $R$ such that every vertex not in $R$ in the last column has at least one neighbor in $R$ and every vertex not in $R$ in the rest of the columns has at least two neighbors in $R$.\\
Clearly, a $2$-dominating set is a quasi-$2$-dominating set such that every vertex in the last column is also $2$-dominated.
\end{definition}

Our objective is to identify each quasi-$2$-dominating set of $C_n\Box P_m$ with a particular vertex labeling that will allow us to handle such vertex subsets with the appropriate algorithms. To this end, let $R$ be a quasi-$2$-dominating set of $C_n\Box P_m$. We label the vertices of $C_n\Box P_m$ with labels $0,1,2$ by using the following rules.
\begin{itemize}
\itemsep=0.95pt
\item every $v\in R$ has label $0$,
\item $v\in V(C_n\Box P_m)\setminus R$ having at least two neighbors in $R$ in its column or in the previous one is labeled as $1$,
\item otherwise, that is, $v\in V(C_n\Box P_m)\setminus R$ has exactly one the neighbor in $R$ in its column or in the previous one, we label the vertex $v$ as $2$.
\end{itemize}

By using this labeling, each column is now identified with an ordered list, that is, a word of length $n$ over the alphabet $\{0,1,2\}$. It is clear that not any word of length $n$ can represent a column of $C_n\Box P_m$, due to the restrictions derived from $R$ being a quasi-$2$-dominating set. By definition of the labeling, subsequences $111, 211, 112, 212$ are not possible because each vertex with label $1$ has at least one neighbor in $R$, that is, labeled as $0$, in its column. Moreover, subsequence $020$ is not allowed by definition of label $2$.

However, vertices in the first column have a different behavior. The first column has no previous one, so subsequences $110, 011, 012, 210$ are not allowed there because both neighbors of vertices with label $1$ in the first column, must be in $R$.

Moreover, in the particular case of a $2$-dominating set $S$, the last column is also different because it has no following one, so every vertex in the last column either belongs to $S$ or has at least two neighbors in $S$ in their column or the previous one. Therefore there is no vertices with label $2$ in the last column.

We resume these ideas in the following definition.

\begin{definition}
Let $\boldsymbol{p}$ be a word of length $n$ over the alphabet $\{0,1,2\}$.
\begin{itemize}
\itemsep=0.95pt
 \leftskip=-2.2mm
\item $\boldsymbol{p}$ is called suitable if it does not contain any of the subsequences $111, 211, 112, 212, 020$.
\item $\boldsymbol{p}$ is called initial if it is suitable and it does not contain any of the subsequences $110, 011, 012, 210$.
\item $\boldsymbol{p}$ is called final if it is suitable and it does not contain any $2$.
\item The weight $\omega(\boldsymbol{p})$ of a suitable word $\boldsymbol{p}$ is the number of $0's$ contained in $\boldsymbol{p}$.
\end{itemize}
\end{definition}
\eject

In the same way that just some words are appropriate to represent a quasi-$2$-dominating set, there are some rules that words in consecutive columns must follow due to the labeling definition and because $R$ is a quasi-$2$-dominating set. We quote such rules in the following definition.

\begin{definition}
Let $\boldsymbol{p}=(p_0, \dots , p_{n-1}), \boldsymbol{q}=(q_0,\dots , q_{n-1})$ be suitable words of length $n$ over the alphabet $\{0,1,2\}$. We say that $\boldsymbol{p}$ can follow $\boldsymbol{q}$ if the following conditions hold, for each $i\in \{0,1,\dots, n-1\}$
 (indices are taken module $n$): \vspace*{-2mm}
\begin{enumerate}
\item[ ] \ \\
\noindent
$
\begin{array}{lll}
  \text{if } q_i=0 \text{ then}&\text{either}& p_i=0 \\
                    &\text{or}& p_i=1\\
                    &\text{or}& p_i=2, p_{i-1}\neq 0, p_{i+1}\neq 0
\end{array}
$

\item[ ] \ \\
\noindent
$
\begin{array}{lll}
  \text{if } q_i=1 \text{ then}&\text{either}& p_i=0 \\
                    &\text{or}& p_i=1,  p_{i-1}= 0, p_{i+1}= 0\\
                    &\text{or}&  p_i=2, p_{i-1}=0\\
                    &\text{or}& p_i=2, p_{i+1}=0
\end{array}
$

\item[ ] \ \\
\noindent
$
\begin{array}{lll}
  \text{if } q_i=2 &\text{then}& p_i=0
\end{array}
$

\end{enumerate}
\end{definition}

The following proposition describes the identification of each quasi-$2$-dominating set of $C_n\Box P_m$ with a particular labeling of the vertices, as we had announced. Such identification will allow us to encode the information of each quasi-$2$-dominating set and to use it in our algorithms.

\begin{proposition}\label{prop:basic}
\begin{enumerate}
\item Each quasi-$2$-dominating set of $C_n\Box P_m$ can be identified with an ordered list $\boldsymbol{p^0}, \boldsymbol{p^1},\dots,$ $\boldsymbol{p^{m-1}}$ of m suitable words of length $n$ such that $\boldsymbol{p^0}$ is initial, and $\boldsymbol{p^{i+1}}$ can follow $\boldsymbol{p^i}$, for every  $i\in \{0,\dots , m-2\}$.

    Conversely, every ordered list $\boldsymbol{p^0}, \boldsymbol{p^1},\dots,$ $\boldsymbol{p^{m-1}}$ of m suitable words of length $n$ such that $\boldsymbol{p^0}$ is initial, and $\boldsymbol{p^{i+1}}$ can follow $\boldsymbol{p^i}$, for every  $i\in \{0,\dots , m-2\}$, represents a unique quasi-$2$-dominating set of $C_n\Box P_m$.

\item If $R=\boldsymbol{p^0}, \boldsymbol{p^0},\dots,$ $\boldsymbol{p^{m-1}}$ is a quasi-$2$-dominating set of $C_n\Box P_m$, then
    $|R|=\sum_{k=0}^{m-1} \omega(\boldsymbol{p^k})$.
\item  $R=\boldsymbol{p^0}, \boldsymbol{p^1},\dots,$ $\boldsymbol{p^{m-1}}$ is a quasi-$2$-dominating set of $C_n\Box P_m$ if and only if $R'=\boldsymbol{p^0}, \boldsymbol{p^1},\dots,$ $\boldsymbol{p^{m-2}}$ is a quasi-$2$-dominating set of $C_n\Box P_{m-1}$ and $\boldsymbol{p^{m-1}}$ can follow $\boldsymbol{p^{m-2}}$.
\item A quasi-$2$-dominating set $S=\boldsymbol{p^0}, \boldsymbol{p^1},\dots,$ $\boldsymbol{p^{m-1}}$ is a $2$-dominating set of $C_n\Box P_m$ if and only if the last word $\boldsymbol{p^{m-1}}$ is final.
\end{enumerate}
\end{proposition}

\begin{proof}

\vspace*{-6mm}
\begin{enumerate}
\item Let $R$ be a quasi-$2$-dominating set of $C_n\Box P_m$ and consider the associated vertex labeling. Let $\boldsymbol{p^j}$ be the word associated to vertices in the $j$-th column, for $j\in \{0,\dots ,m-1\}$. Then, the rules of the labeling ensure that the word list $\boldsymbol{p^0}, \boldsymbol{p^1},\dots,$ $\boldsymbol{p^{m-1}}$ satisfies the desired properties.

    Conversely, consider a word list $\boldsymbol{p^0}, \boldsymbol{p^1},\dots,$ $\boldsymbol{p^{m-1}}$ with the properties described above. For each $j\in \{0,\dots, m-1\}$, $\boldsymbol{p^j}=(p^j_1, \dots , p^j_{n-1})$ and we label the vertices in the $j$-th column as $v_{ij}=p^j_i, i\in \{1,\dots , n-1\}$. Clearly, the set of all vertices in $C_n\Box P_m$ with label $0$ is quasi-$2$-dominating.

\item Let $R=\boldsymbol{p^0}, \boldsymbol{p^1},\dots,$ $\boldsymbol{p^{m-1}}$ be a quasi-$2$-dominating set, identified with its word list. By the construction provided in the preceding item, the cardinal of $R$ is the number of vertices with label $0$, that is $|R|=\sum_{k=0}^{m-1} \omega(\boldsymbol{p^k})$.

\item Let $R=\boldsymbol{p^0}, \boldsymbol{p^1},\dots,$ $\boldsymbol{p^{m-1}}$ be a quasi-$2$-dominating set. Then, every vertex in columns from the first one to the $(m-2)$-th has at least two neighbors in $R$, that are also in $R'=\boldsymbol{p^0}, \boldsymbol{p^1},\dots,$ $\boldsymbol{p^{m-2}}$, except perhaps in the case of column $(m-2)$-th, where just one neighbor in $R'$ is ensured.  Moreover $\boldsymbol{p^0}$ is initial and $\boldsymbol{p^{i+1}}$ can follow $\boldsymbol{p^i}$, for every  $i\in \{0,\dots , m-2\}$. This means that $R'$ is a quasi-$2$-dominating set of $C_n\Box P_{m-1}$, as desired.\vspace{1mm}

    Conversely, let $\boldsymbol{p^0}, \boldsymbol{p^1},\dots,$ $\boldsymbol{p^{m-1}}$ be a word list such that $R'=\boldsymbol{p^0}, \boldsymbol{p^1},\dots,$ $\boldsymbol{p^{m-2}}$ is a quasi-$2$-dominating set of $C_n\Box P_{m-1}$ and $\boldsymbol{p^{m-1}}$ can follow $\boldsymbol{p^{m-2}}$. Then, vertices in columns from the first one to the $(m-3)$-th have at least two neighbors with label $0$, because $R'$ is quasi-$2$-dominating, vertices in the $(m-2)$-th column have at least two neighbors with label $0$ because $R'$ is quasi-$2$-dominating and $\boldsymbol{p^{m-1}}$ can follow $\boldsymbol{p^{m-2}}$ and vertices in $(m-1)$-th column have at least one neighbor with label $0$, because $\boldsymbol{p^{m-2}}$ is suitable and $\boldsymbol{p^{m-1}}$ can follow $\boldsymbol{p^{m-2}}$. These conditions ensure that $R=\boldsymbol{p^0}, \boldsymbol{p^1},\dots,$ $\boldsymbol{p^{m-1}}$ is a quasi-$2$-dominating set of $C_n\Box P_m$.

\item Clearly, a quasi-$2$-dominating set $S=\boldsymbol{p^0}, \boldsymbol{p^1},\dots,$ $\boldsymbol{p^{m-1}}$ is a $2$-dominating set if and only if every vertex in the last column either is in $S$ or has at least two neighbors in $S$, so $S$ is a $2$-dominating set if and only if $\boldsymbol{p^{m-1}}$ does not contain any $2$, that is, it is a final word.
\end{enumerate}

\vspace*{-6mm}
\end{proof}

We now define the tools that will be needed to compute the $2$-domination number of the cylinder $C_n\Box P_m$. Let $n\geq 3$ be an integer and denote by $s(n)$ the number of suitable words of length $n$. The initial vector $X^1=(X^1(\boldsymbol{p^1}), \dots ,$
$X^1(\boldsymbol{p^{s(n)}}))$ is a vector of length $s(n)$, such that each entry corresponds to a suitable word, defined as follows.
\begin{eqnarray} \label{equation:vector}
X^1(\boldsymbol{p})=
\left\{
\begin{array}{ll}
\omega(\boldsymbol{p}) & \text{if } \boldsymbol{p} \text{ is initial}\\
\infty & \text{otherwise}
\end{array}
\right.
\end{eqnarray}

The transition matrix $A=(A_{\boldsymbol{pq}})$ is the square matrix with size $s(n)$ such that each entry corresponds to a pair of suitable words, defined as follows.
\begin{eqnarray}\label{equation:matrix}
A_{\boldsymbol{pq}}=
\left\{
\begin{array}{ll}
\omega(\boldsymbol{p}) & \text{if } \boldsymbol{p} \text{ can follow } \boldsymbol{q}\\
\infty & \text{otherwise}
\end{array}
\right.
\end{eqnarray}

In the following theorem we describe how the initial vector and the transition matrix collect the information about the quasi-$2$-dominating sets of a cylinder needed to compute its $2$-domination number.

\begin{theorem}\label{th:vector}

Let $n\geq 3$ be an integer. Let $X^1$ be the initial vector defined by Equation~\ref{equation:vector} and let $A$ be the transition matrix defined by Equation~\ref{equation:matrix}. Let $m\geq 2$ be an integer and let $X^2, \dots , X^m$ be the vectors recursively obtained by $X^{i+1}=A\boxtimes X^i$. Then
\begin{eqnarray*}
X^m(\boldsymbol{p})=
\left\{
{\small
\begin{array}{ll}
\infty & \text{\small if there is no quasi-2-dominating set in } \\
 & C_n\Box P_m\text{ with word } \boldsymbol{p} \text{ in the last column}\\
 & \\
\text{minimum cardinal of a quasi-2-dominating set } & \\
\text{ of } C_n\Box P_m \text{ with word } \boldsymbol{p} \text{ in the last column} & \text{otherwise}
\end{array}
}
\right.
\end{eqnarray*}
\end{theorem}

\begin{proof}
We proceed by induction over $m\geq 2$. First of all, $X^2=A\boxtimes X^1$, so for each suitable word $\boldsymbol{p}$, by definition of the $(\min ,+)$ matrix multiplication, we obtain:
\begin{eqnarray*}
X^2(\boldsymbol{p})=\min \{ A_{\boldsymbol{p p^k}}+X^1(\boldsymbol{p^k}) \colon k\in \{1,\dots , s(n)\} \}
\end{eqnarray*}

We now consider two different cases.
\par\bigskip
Case 1: $X^2(\boldsymbol{p})=\infty$.

This means that, for every  $k\in \{1,\dots , s(n)\}$ either $A_{\boldsymbol{p p^k}}=\infty$ or $X^1(\boldsymbol{p^k})=\infty$, that is, either $\boldsymbol{p}$ cannot follow $\boldsymbol{p^k}$ or $\boldsymbol{p^k}$ is not initial. Therefore, there is no quasi-$2$-dominating set in $C_n\Box P_2$ with $\boldsymbol{p}$ in the last column.
\par\bigskip

Case 2: $X^2(\boldsymbol{p})<\infty$.

Then, there exists an index $\ell$ such that
\begin{eqnarray*}
\min \{ A_{\boldsymbol{p p^k}}+X^1(\boldsymbol{p^k}) \colon k\in \{1,\dots , s(n)\} \}=A_{\boldsymbol{p p^\ell}}+X^1(\boldsymbol{p^\ell})<\infty
\end{eqnarray*}
\noindent and, in particular, both $A_{\boldsymbol{p p^\ell}}$ and $X^1(\boldsymbol{p^\ell})$ are finite.

\medskip
Therefore, $\boldsymbol{p^\ell}$ is an initial word, $X^1(\boldsymbol{p^\ell})= \omega( \boldsymbol{p^\ell})$,  the word $\boldsymbol{p}$ can follow $\boldsymbol{p^\ell}$ and $A_{\boldsymbol{p p^\ell}}=\omega( \boldsymbol{p})$. So the quasi-$2$-dominating set represented by the word list $\boldsymbol{p^\ell p}$ has the word $\boldsymbol{p}$ in the last column and, by using Proposition~\ref{prop:basic}, its  cardinal is $\omega(\boldsymbol{p}) +  \omega(\boldsymbol{p^\ell})=A_{\boldsymbol{p p^\ell}}+X^1(\boldsymbol{p^\ell})$, which is minimum among the cardinals of the quasi-$2$-dominating sets having $\boldsymbol{p}$ in the last column, by the election of the index $\ell$. This concludes the first case of the induction and we now proceed with the inductive step.
\par\bigskip

Assume that the vector $X^{m-1}$ satisfies the desired properties, that is,
\begin{eqnarray*}
X^{m-1}(\boldsymbol{p})=
\left\{
{\small
\begin{array}{ll}
\infty & \text{if there is no quasi-2-dominating} \\
 & \text{set in } C_n\Box P_{m-1}\text{ with the word } \boldsymbol{p}\\
 &\text{in the last column}\\
 & \\
\text{minimum cardinal of a quasi-2-dominating } & \\
\text{set of } C_n\Box P_{m-1} \text{ with the word } \boldsymbol{p} \text{ in the}&\\
\text{last column} & \text{otherwise}
\end{array}
}
\right.
\end{eqnarray*}
\noindent and let $X^m=A\boxtimes X^{m-1}$. Then, $X^m(\boldsymbol{p})=\min \{ A_{\boldsymbol{p p^k}}+X^{m-1}(\boldsymbol{p^k}) \colon k\in \{1,\dots , s(n)\} \}$ and we again consider two cases.

\eject

Case 1: $X^m(\boldsymbol{p})=\infty$.

Then, $A_{\boldsymbol{p p^k}}+X^{m-1}(\boldsymbol{p^k})=\infty$ for every $k\in \{1,\dots , s(n)\} $. Suppose, on the contrary, that there is a quasi-$2$-dominating set in $C_n\Box P_{m}$ with the word $\boldsymbol{p}$ in the last column, represented by the word list $ \boldsymbol{p^{k_1}}, \dots ,$ $\boldsymbol{p}^{\boldsymbol{k_{m-1}}}$, $\boldsymbol{p}$. Thus, the word $\boldsymbol{p}$ can follow the word $\boldsymbol{p}^{\boldsymbol{k_{m-1}}}$ and moreover, by using Proposition~\ref{prop:basic}, the word list $ \boldsymbol{p^{k_1}}, \dots, $ $\boldsymbol{p}^{\boldsymbol{k_{m-1}}}$ represents to a quasi-$2$-dominating set in $C_n\Box P_{m-1}$ with $\boldsymbol{p}^{\boldsymbol{k_{m-1}}}$ in the last column. Both conditions give that $A_{\boldsymbol{p} \boldsymbol{p}^{\boldsymbol{k_{m-1}}}} +X^{m-1}(\boldsymbol{p}^{\boldsymbol{k_{m-1}}})<\infty$, a contradiction to the hypothesis of this case.
\par\bigskip

Case 2: $X^m(\boldsymbol{p})<\infty$.

Then, there exists an index $\ell$ such that $X^m(\boldsymbol{p})=\min \{ A_{\boldsymbol{p p^k}}+X^{m-1}(\boldsymbol{p^k}) \colon k\in \{1,\dots , s(n)\}=A_{\boldsymbol{p p^\ell}}+X^{m-1}(\boldsymbol{p^\ell})<\infty$. Thus, by the inductive hypothesis, there exists a quasi-$2$-dominating set in $C_n\Box P_{m-1}$ with word list $ \boldsymbol{p^{k_1}}, \dots ,$ $\boldsymbol{p}^{\boldsymbol{k_{m-2}}}$, $\boldsymbol{p^\ell}$ and, by definition of the transition matrix $A$, $\boldsymbol{p}$ can follow $\boldsymbol{p^\ell}$.

\medskip
Finally, again by using  Proposition~\ref{prop:basic}, the word list  $ \boldsymbol{p^{k_1}}, \dots ,$ $\boldsymbol{p}^{\boldsymbol{k_{m-2}}}$, $\boldsymbol{p^\ell},$ $\boldsymbol{p}$ represents a quasi-$2$-dominating set in $C_n\Box P_{m}$, with $\boldsymbol{p}$ in the last column and having cardinal $A_{\boldsymbol{p p^\ell}}+X^{m-1}(\boldsymbol{p^\ell})=X^m(\boldsymbol{p})$, which is the minimum among the cardinal of all the quasi-$2$-dominating sets satisfying the same conditions, by the election of the index $\ell$.

\medskip
This concludes the inductive proof.
\end{proof}

We conclude the results that allow us to compute $\gamma_2(C_n\Box P_m)$, where $n$ and $m$ are fixed integers, with the following theorem that will provide theoretical support for the algorithmic results.

\begin{theorem}\label{th:two_parameters}
Let $n\geq 3$ and $m\geq 2$ be integers. Let $X^1$ be the initial vector defined by Equation~\ref{equation:vector} and let $A$ be the transition matrix defined by Equation~\ref{equation:matrix}.  Let $X^2, \dots , X^m$ be the vectors recursively obtained by $X^{i+1}=A\boxtimes X^i$. Then
\begin{eqnarray*}
\gamma_2(C_n\Box P_m)=\min \{ X^m(\boldsymbol{p})\colon \boldsymbol{p} \text{ is a final world}\}
\end{eqnarray*}
\end{theorem}

\begin{proof}
Let $S'$ be a $2$-dominating set in $C_n\Box P_m$ such that $\gamma_2(C_n\Box P_m)=|S'|$. Then, the last word of $S'$, say $\boldsymbol{p'}$, is a final word and, by Theorem~\ref{th:vector},
\begin{eqnarray*}
  |S'|&=&\text{minimum cardinal of a quasi-2-dominating set of } C_n\Box P_{m}\\
      &=&\text{minimum cardinal of a quasi-2-dominating set of } C_n\Box P_{m} \text{ with word } \boldsymbol{p'} \text{ in the last column} \\
      &=& X^m(\boldsymbol{p'})
\end{eqnarray*}

However, by the selection of $S'$, it is clear that
$X^m(\boldsymbol{p'})=\min \{ X^m(\boldsymbol{p})\colon \boldsymbol{p} \text{ is a final word}\}$, so finally,
\begin{eqnarray*}
\gamma_2(C_n\Box P_m)&=&|S'|\\
                     &=&X^m(\boldsymbol{p'})\\
                     &=&\min \{ X^m(\boldsymbol{p})\colon \boldsymbol{p} \text{ is a final word}\}
\end{eqnarray*}

\vspace*{-7mm}
\end{proof}

The last theorem gives a procedure to compute $\gamma_2(C_n\Box P_m)$, for fixed $n$ and $m$.
We finish this section with a standard argument about the $(\min,+)$ matrix multiplication, that allows us to obtain $\gamma_2(C_n\Box P_m)$, for fixed $n$ and any $m\geq 2$. To this end, we need the following well known lemma whose proof we include here for the sake of completeness.

\begin{lemma}\label{lem:recurrence}
Let $A$ be a square matrix and let $X^1$ be a vector with the same size as $A$. Let $X^i$ be the vectors recursively obtained by $X^{i+1}=A\boxtimes X^i$ for $i\geq 2$ and suppose that there exist natural numbers $m_0,a,b$ such that $X^{m_0+a}=b\boxtimes X^{m_0}$. Then, $X^{m+a}=b\boxtimes X^{m}$, for every $m\geq m_0$.
\end{lemma}

\begin{proof}
We proceed by induction. By hypothesis, $X^{m_0+a}=b\boxtimes X^{m_0}$. Let $m\geq m_0$ be such that $X^{m+a}=b\boxtimes X^{m}$ then, $X^{(m+1)+a}=A\boxtimes X^{m+a}=A\boxtimes(b\boxtimes X^{m})=b\boxtimes(A\boxtimes X^{m})=b\boxtimes X^{m+1}$, as desired.
\end{proof}

\begin{theorem}\label{th:one_parameter}
Let $n\geq 3$ be an integer. Let $X^1$ be the initial vector defined by Equation~\ref{equation:vector} and let $A$ be the transition matrix defined by Equation~\ref{equation:matrix}. Let $X^i$ be the vectors recursively obtained by $X^{i+1}=A\boxtimes X^i$ for $i\geq 1$ and assume that there exist integers $m_0,a,b$ such that $X^{m_0+a}=b\boxtimes X^{m_0}$. Then, $\gamma_2(C_n\Box P_{m+a})=b+\gamma_2(C_n\Box P_m)$, for every $m\geq m_0$.
\end{theorem}

\begin{proof}
By using Theorem~\ref{th:two_parameters} and Lemma~\ref{lem:recurrence}, we obtain that
\begin{eqnarray*}
\gamma_2(C_n\Box P_{m+a})&=&\min \{ X^{m+a}(\boldsymbol{p})\colon \boldsymbol{p} \text{ is a final word}\}\\
                         &=&\min \{ b\boxtimes X^{m}(\boldsymbol{p})\colon \boldsymbol{p} \text{ is a final word}\}\\
                         &=& \min \{ b+X^m(\boldsymbol{p})\colon \boldsymbol{p} \text{ is a final word}\}\\
                         &=& b+\min \{X^m(\boldsymbol{p})\colon \boldsymbol{p} \text{ is a final word}\}\\
                         &=& b+\gamma_2(C_n\Box P_{m})
\end{eqnarray*}

\vspace*{-7mm}
\end{proof}

Note that, given a fixed integer $n\geq 3$, the theorem above provides the following finite difference equation for the $2$-dominating number of the cylinder $C_n\Box P_m$, where $a$, $b$, and $m_0$ are as in Theorem~\ref{th:one_parameter}.
$$\gamma_2(C_n\Box P_{m+a})-\gamma_2(C_n\Box P_{m}))=b, \text{ for } m\geq m_0$$
whose boundary values are $\gamma_2(C_n\Box P_{m})$ for $m_0\leq m\leq m_0+a-1$. The unique solution of such finite difference equation provides $\gamma_2(C_n\Box P_m)$ for every $m\geq m_0$. In addition, the procedure described in Theorem~\ref{th:two_parameters} gives the remaining values, that is, the values of $\gamma_2(C_n\Box P_m)$ for $2\leq m\leq m_0-1$.

\section{Algorithmic results}\label{Sec:algorithms}

In this section we present the algorithms that we have used to compute $\gamma_2(C_n\Box P_m)$, in cases $3\leq n \leq 15$. We have run the algorithms in a CPU AMD EPYC 7642  and we also present the results that we have obtained. We have divided the computation into two steps. Firstly, Algorithm~\ref{alg:recurrence} provides the finite difference equation presented in Theorem~\ref{th:one_parameter}.
\begin{algorithm}[h]
\caption{Searching for the finite difference equation for $\gamma _2(C_n\Box P_{m})$, with $n$ fixed}
\label{alg:recurrence}
\LinesNumbered
\KwIn{$n\geq 3$, a natural number}
\KwOut{the finite difference equation $\gamma _2(C_n\Box P_{m+a})-\gamma_2(C_n\Box P_m)=b$, for $m\geq m_0$
or finite difference equation not found }
compute all suitable words of length $n$\;
compute the initial words and the initial vector $X^1$\;
compute matrix $A$\;
compute the vectors  $X^{i+1}=A\boxtimes X^i$, for $i\leq K$ big enough\;
 {\If {$X^{m_0+a}=b \boxtimes X^{m_0}$ \text{ for natural numbers} $m_0,a,b$}
  {\Return $m_0, a, b$}
  {\Else
  {\Return finite difference equation not found}
  }
  }
\end{algorithm}

There are some sufficient conditions to ensure that Step 5  of Algorithm~\ref{alg:recurrence} is true (see~\cite{Spalding1998}) but they provide a huge value for $m_0$, in the order of the square of the matrix size, and they are not practical. We have looked for the desired relationship just by checking the vectors computed in Step 4, with $K=20$.

In addition, from the computational point of view it is more efficient to compute vectors in Step 4 recursively as $X^{i+1}=A\boxtimes X^i$, for $i\leq 20$, instead of using the alternative formula $X^i=A^{i-1}\boxtimes X^1$. The reason is that the initial matrix $A$ is sparse but it becomes dense after a small number ($3$ or $4$) of $(\min,+)$ powers, so the resources needed to compute and store the matrix $A$ and $20$ vectors are smaller than the resources needed to compute and store $20$ successive powers of $A$. The $(\min,+)$ product of the matrix $A$ and the vectors $X^i$ has been done with a modification of the library CSPARSE~\cite{CSparse}, to adapt it to such product. The computation times of Steps 1 and 2 in Algorithm~\ref{alg:recurrence} are negligible compared to those of Steps 3 and 4, even in the largest case $n=15$, so we do not show them. We show the results obtained by Algorithm~\ref{alg:recurrence} in Table~\ref{table:values}.

\begin{table}[!h]
\centering
\caption{Results obtained by Algorithm~\ref{alg:recurrence}}%
\setlength{\tabcolsep}{3pt}
{\small
\begin{tabular}{>{\raggedleft\arraybackslash}p{10pt} >{\raggedleft\arraybackslash}p{60pt}> {\raggedleft\arraybackslash}p{55pt}>{\raggedleft\arraybackslash}p{70pt}>{\raggedleft\arraybackslash}p{80pt} >{\raggedleft\arraybackslash}p{25pt}
>{\raggedleft\arraybackslash}p{25pt} >{\raggedleft\arraybackslash}p{25pt}}
\hline
$n$ & Number of suitable words & Memory size of the matrix $A$ & Computation time of the matrix $A$ & Computation time of vectors $X^i, i\leq 20$ &$m_0$ & $a$ & $b$ \\
\hline
3& 17 & $0.32 KB$ & $< 1 s.$ & $< 1 s.$ & 5& 1 &1 \\
\hline
4& 40 & $1.18 KB$ & $< 1 s.$& $< 1 s.$ &6& 2 &3 \\
\hline
5& 92 & $4.64 KB$ &$< 1 s.$ &$< 1 s.$ & 8& 2& 4 \\
\hline
6& 235 & $19.82 KB$ &$< 1 s.$ &$< 1 s.$ & 7 & 1 & 2 \\
\hline
7& 590 &   $82.34 KB$   &$< 1 s.$ &$< 1 s.$ & 8 & 2 & 5\\
\hline
8& 1456 &   $339.27 KB$   &$< 1 s.$ &$< 1 s.$ &7 & 2 & 6\\
\hline
9& 3617 &  $1.37 MB$   &  $< 1 s.$& $3 s.$   &8 & 1& 3  \\
\hline
10& 9004 &   $5.70 MB$  & $4 s. $&   $10 s. $ &9 & 2 & 7\\
\hline
11& 22376 & $23.61 MB$    & $22 s. $&  $ 43 s.$  &10 & 2 & 8\\
\hline
12& 55603 &  $97.84 MB$   & $2 m. 14 s. $& $2 m. 53 s. $   &11 & 1 & 4 \\
\hline
13& 138218 &  $405.51 MB$  &$  13 m. 30 s.$ & $11m. 50 s. $  &10 & 2 & 9\\
\hline
14& 343564 &  $0.95GB$   & $ 82 m. 10 s. $ & $49m. 13s. $  &11 & 2 & 10\\
\hline
15& 853937 &  $6.8GB$   & $ 8h. 23 m. 34s. $ &  $3 h. 24 m. 51 s.$ & 11&1  &5 \\
\hline
\end{tabular}
}
\label{table:values}
\end{table}

\medskip	
The values in Table~\ref{table:values} provide the finite difference equation $f(m+a)-f(m)=b, m\geq m_0$, where $f(m)=\gamma_2(C_n\Box P_m)$, for fixed $n\in \{3, \dots ,15\}$. Note that a regular behavior can be found in the values of $a$ and $b$, for the cases $3\leq n\leq 15$.
\begin{itemize}
\item If $n\equiv 0\pmod 3$ then, $a=1$ and $\displaystyle b=\frac{n}{3}\cdot$
\item If $n\equiv 1\pmod 3$ then, $a=2$ and $\displaystyle b=\frac{2n+1}{3}\cdot$
\item If $n\equiv 2\pmod 3$ then, $a=2$ and  $\displaystyle b=\frac{2n+2}{3}\cdot$
\end{itemize}

The boundary values of the finite difference equations provided above are $f(m)=\gamma_2(C_n\Box P_{m})$ for $m_0\leq m\leq m_0+a-1$. We have computed them by using Algorithm~\ref{alg:small}, which follows from the result shown in Theorem~\ref{th:two_parameters}.

We have also computed with Algorithm~\ref{alg:small} the remaining values, that is, $\gamma_2(C_n\Box P_m)$ for every $2\leq m\leq m_0-1$. All the values are in Table~\ref{table:equation}.

\medskip
\begin{algorithm}[!h]
\LinesNumbered
\caption{Computation of $\gamma_2(C_n\Box P_m)$, for $n,m$ fixed}
\label{alg:small}
\KwIn{$n\geq 3, m\geq 2$, natural numbers}
\KwOut{$\gamma_2(C_n\Box P_m)$}
compute all suitable words of length $n$\;
compute the initial words and the initial vector $X^1$\;
compute matrix $A$\;
compute the vectors  $X^{i+1}=A\boxtimes X^i$, for $1\leq i\leq m-1$\;
compute the final words\;
\Return $\min \{ X^m(\boldsymbol{p})\colon \boldsymbol{p} \text{ is a final word}\}$
\end{algorithm}\medskip

\begin{table}[!h]
 {\scriptsize
\centering
\caption{Boundary values and remaining values obtained with Algorithm~\ref{alg:small}}%
\setlength{\extrarowheight}{0.1cm}
\setlength{\tabcolsep}{3pt}
\scalebox{0.98}{\hspace{2mm}\begin{tabular}{>{\raggedleft\arraybackslash}p{10pt} >{\raggedright\arraybackslash}p{120pt} >{\raggedright\arraybackslash}p{90pt}
>{\raggedright\arraybackslash}p{200pt}}
\hline
{\footnotesize $n$} & {\footnotesize finite difference equation} & {\footnotesize boundary values} & {\footnotesize remaining values}\\
\hline
3& $f(m+1)-f(m)=1, m \geq 5$ & $f(5)=7$& $f(2)=3, f(3)=4, f(4)=6$  \\
\hline
4&  $f(m+2)-f(m)=3, m \geq 6$ & $f(6)=11, f(7)=12$ & $f(2)=4, f(3)=6, f(4)=8, f(5)=9$  \\
\hline
5& $f(m+2)-f(m)=4, m\geq 8$ & $f(8)=18, f(9)=19$ & $f(2)=5$, $f(3)=7$, $f(4)=10$, $f(5)=11$, $f(6)=14$, $f(7)=15 $  \\
\hline
6& $f(m+1)-f(m)=2, m\geq 7$ & $f(7)=18$ & $f(2)=6$, $f(3)=8$, $f(4)=11$, $f(5)=13$, $f(6)=16$\\
\hline
7& $f(m+2)-f(m)=5, m\geq 8$ &  $f(8)=24, f(9)=26$ & $f(2)=7$, $f(3)=10$, $f(4)=13$, $f(5)=15$, $f(6)=18$, $f(7)=21$  \\
\hline
8& $f(m+2)-f(m)=6, m\geq 7$ &  $f(7)=24, f(8)=27$  &$f(2)=8$, $f(3)=11$, $f(4)=14$, $f(5)=18$, $f(6)=21$\\
\hline
9& $f(m+1)-f(m)=3, m\geq 8$ &  $f(8)=30$ & $f(2)=9$, $f(3)=12$, $f(4)=16$, $f(5)=20$, $f(6)=24$, $f(7)=27$  \\
\hline
10& $f(m+2)-f(m)=7, m\geq 9$ & $f(9)=37, f(10)=41$ & $f(2)=10$, $f(3)=14$, $f(4)=18$, $f(5)=22$, $f(6)=26$, $f(7)=30$, $f(8)=34$ \\
\hline
11& $f(m+2)-f(m)=8, m\geq 10$ & $f(10)=45, f(11)=49$  & $f(2)=11$, $f(3)=15$, $f(4)=20$, $f(5)=24$, $f(6)=28$, $f(7)=33$, $f(8)=37$, $f(9)=41$ \\
\hline
12& $f(m+1)-f(m)=4, m\geq 11$  & $f(11)=52$ & $f(2)=12$, $f(3)=16$, $f(4)=22$, $f(5)=26$, $f(6)=31$, $f(7)=36$, $f(8)=40$, $f(9)=44$, $f(10)=48$ \\
\hline
13& $f(m+2)-f(m)=9, m\geq 10$ & $f(10)=53, f(11)=57$ & $f(2)=13$, $f(3)=18$, $f(4)=24$, $f(5)=28$, $f(6)=34$, $f(7)=39$, $f(8)=44$, $f(9)=48$ \\
\hline
14& $f(m+2)-f(m)=10, m\geq 11$ & $f(11)=62, f(12)=67$ & $f(2)=14$, $f(3)=19$, $f(4)=25$, $f(5)=30$, $f(6)=36$, $f(7)=42$, $f(8)=47$, $f(9)=52$, $f(10)=57$ \\
\hline
15& $f(m+1)-f(m)=5, m\geq 11$  & $f(11)=65$ &  $f(2)=15$, $f(3)=20$, $f(4)=27$, $f(5)=33$, $f(6)=39$, $f(7)=45$, $f(8)=50$, $f(9)=55$, $f(10)=60$\\
\hline
\end{tabular} }
\label{table:equation}
}
\end{table}

Again, the running times of Steps 5 and 6 in Algorithm~\ref{alg:small} are negligible compared to those of Steps 3 and 4, even in the largest case $n=15$.
With the data in Tables~\ref{table:values} and~\ref{table:equation}, we have solved the finite difference equation $\gamma_2(C_n\Box P_{m+a})-\gamma_2(C_n\Box P_{m}))=b, \text{ for } m\geq m_0$, for each value of n, and therefore, we have computed the general formula of $\gamma_2(C_n\Box P_{m})$, with $3\leq n\leq 15$ and $m\geq m_0$. In some cases, some remaining values ($m<m_0$) also agree with the solution of the equation and we have included them. The values of $\gamma_2(C_n\Box P_{m})$ with $3\leq n\leq 15$, $m<m_0$ and such that they do not follow the general formula, can be found in Table~\ref{table:equation}.
\par\bigskip

$\gamma_2(C_3\Box P_{m})=m+2$, for $m\geq 4$
\par\bigskip

$\gamma_2(C_4\Box P_{m})=\displaystyle \left\lceil \frac{3m+3}{2}\right\rceil$, for $m\geq 3$
\par\bigskip

$\gamma_2(C_5\Box P_{m})=
\left\{
\renewcommand{\arraystretch}{1.3}
\begin{array}{ll}
2m+2& \text{if } m\neq 2 \text{ and } m\equiv 0\pmod 2 \\
2m+1& \text{if } m=2 \text{ or } m\equiv 1\pmod 2
\end{array}
\right.
$
\par\bigskip

$\gamma_2(C_6\Box P_{m})=2m+4$, for $m\geq 6$
\par\bigskip

$\gamma_2(C_7\Box P_{m})=\displaystyle \left\lceil \frac{5m+7}{2}\right\rceil$, for $m\geq 7$
\par\bigskip

$\gamma_2(C_8\Box P_{m})=3m+3$, for $m\geq 5$

\par\bigskip

$\gamma_2(C_9\Box P_{m})=3m+6$, for $m\geq 6$
\par\bigskip

$\gamma_2(C_{10}\Box P_{m})=\displaystyle \left\lceil \frac{7m+11}{2}\right\rceil$, for $m\geq 7$

\par\bigskip

$\gamma_2(C_{11}\Box P_{m})=4m+5$, for $m\geq 7$

\par\bigskip

$\gamma_2(C_{12}\Box P_{m})=4m+8$, for $m\geq 7$

\par\bigskip

$\gamma_2(C_{13}\Box P_{m})=\displaystyle \left\lceil \frac{9m+15}{2}\right\rceil$, for $m\geq 7$

\par\bigskip

$\gamma_2(C_{14}\Box P_{m})=5m+7$, for $m\geq 7$

\par\bigskip

$\gamma_2(C_{15}\Box P_{m})=5m+10$, for $m\geq 7$

\par\bigskip

\section{Conclusions}\label{sec:conclusions}
In this paper we present two algorithms to compute the $2$-domination number of cylinders $C_n\Box P_m$ with $3\leq n\leq 15$ and $m\geq 2$. Both algorithms are based on the $(\min,+)$ matrix multiplication and they adapt to this graph family and this parameter a technique that has been used to compute other domination parameters in the Cartesian product of two paths and a cycle and a path. The first algorithm provides a finite difference equation for the $2$-domination number of $C_n\Box P_m$, where $n$ is fixed and small enough and the second one computes the boundary values of such equation and the remaining values not included in it.

The size of the matrices used by both algorithms grows exponentially with $n$ and it is the key to decide if a value of $n$ is suitable or not. In our case, we have run the algorithms for $n\leq 15$ and in the largest case the matrix has a size of 853937 rows and columns. From the computational point of view, we have taken advantage of the fact that such matrices are sparse and we have used a modification of the library CSPARSE, to compute the $(\min,+)$ product of the matrices and appropriate vectors.

\medskip
The algorithms have provided formul\ae\ for $\gamma_2(C_{n}\Box P_{m}), 3\leq n\leq 15, m\geq 2$ and they follow regular patterns, for big enough values of $m$, except for the case $n=5$:
\begin{itemize}
\item If $n=3,6,9,12,15$ then, $\displaystyle\gamma_2(C_{n}\Box P_{m})=\frac{n(m+2)}{3}\cdot$
\item If $n=4,7,10,13$ then, $\displaystyle\gamma_2(C_{n}\Box P_{m})=\left\lceil\frac{\frac{2n+1}{3}(m+1)}{2}+\frac{n-4}{3}\right\rceil\cdot$
\item If $n=8,11,14$ then, $\displaystyle\gamma_2(C_{n}\Box P_{m})=\frac{(n+1)(m+1)}{3}+\frac{n-8}{3}\cdot$
\end{itemize}

We think that these formul\ae\  could be also valid for larger values of $n$.

\subsection*{Acknowledgements}

This work was partially supported by the following grants of the Spanish Ministry of Science and Innovation:
RTI2018-095993-B-I00 and PID2019-104129GB-I00/AEI/10.13039/501100011033, both funded by
MCIN/AEI/10.13039/ 501100011033/FEDER ``A way to make Europe.''

\bibliographystyle{fundam}
\bibliography{bibfile}

\end{document}